\documentclass{article}
\usepackage{fullpage}
\usepackage{amsmath}
\usepackage{amsthm}
\usepackage{amssymb}
\usepackage{url}
\usepackage{graphicx}
\usepackage{verbatim}
\usepackage{listings}
\usepackage{url}
\usepackage{graphicx}
\usepackage{tikz}
\usepackage{algorithmic}
\usepackage{algorithm}
\usepackage{subfig}
\usetikzlibrary{calc}
\usetikzlibrary{patterns}
\tikzstyle{mybox} = [draw=black, fill=white,  thick,
    rectangle, inner sep=10pt, inner ysep=20pt]
\tikzstyle{mybox} = [draw=black, fill=white,  thick,
    rectangle, inner sep=2pt, inner ysep=2pt]

\newtheorem{thm}{Theorem}
\newtheorem{lemma}{Lemma}

\theoremstyle{definition}

\begin{document}


\title{Fast Approximation and Randomized Algorithms for Diameter}
\author{Sharareh Alipour\thanks{Department of Computer Engineering,
        Sharif University of Technology, Tehran, Iran. {\tt shalipour@ce.sharif.edu}, This work was completed in part when the author visited DIMACS as an REU mentor in 2014}~~and~~Bahman Kalantari\thanks{Department of Computer Science,
        Rutgers University, New Brunswick, New Jersey, USA. {\tt kalantari@cs.rutgers.edu}}
        ~~and~~ Hamid Homapour\thanks{Department of Computer Engineering,
        Sharif University of Technology, Tehran, Iran. {\tt homapour@ce.sharif.edu}}.}
\date{}
\maketitle



\begin{abstract}
We consider approximation of diameter of a set $S$ of $n$ points in dimension $m$. E$\tilde{g}$ecio$\tilde{g}$lu and Kalantari \cite{kal} have shown that given any $p \in S$, by computing its farthest in $S$, say $q$, and in turn the farthest point of $q$, say $q'$, we have ${\rm diam}(S) \leq \sqrt{3}~ d(q,q')$.  Furthermore, iteratively replacing $p$ with an appropriately selected point on the line segment $pq$, in at most $t \leq n$ additional iterations, the constant bound factor is improved  to $c_*=\sqrt{5-2\sqrt{3}} \approx 1.24$. Here we prove when $m=2$, $t=1$.  This suggests in practice a few iterations may produce good  solutions in any dimension.  Here we also propose a randomized version and present large scale computational results with these algorithm for arbitrary $m$. The algorithms outperform many existing algorithms.  On sets of data as large as $1,000,000$ points, the proposed algorithms compute solutions to within an absolute error of $10^{-4}$.
\end{abstract}


\section{Introduction}
Given a finite set of points $S$ in $\mathbb{R} ^m$, the diameter of $S$,
denoted by ${\rm diam}(S)$,  is defined as the maximum distance between two points of $S$.  Yao \cite{yao}  has considered the case of $m >2$. For $m=2$ the problem can be solved in $O(n\log n)$ time. Computing the diameter of a point set is a fundamental problem and has a long history.
It can be shown that computing the diameter of $n$ points in $\mathbb{R} ^m$ requires $\Omega(n \log n)$ operations in the algebraic computation-tree model \cite{prep}. The problem becomes much harder in $R^3$. Clarkson and Shor gave a randomized $O(n\log n)$ algorithm\cite{cla}.
Recent attempts to solve the 3-dimensional diameter problem led to $O(n\log^3n)$ \cite{ram1,agr} and $O(n\log^2n)$ deterministic algorithms \cite{ram1,bes}.
Finally Ramos found an optimal $O(n \log n)$ deterministic algorithm\cite{ram2}.
All these algorithms use complex data structures and algorithmic techniques such as 3-dimensional convex hulls, intersection of balls, furthest-point Voronoi diagrams, point location search structures or parametric search.
There are many other papers that focus on this problem, see \cite{fou,agha,chan,mal,her}. The first nontrivial approximation algorithm for this problem for arbitrary $m$  was given in \cite{kal}, approximating the diameter to within a factor of $\sqrt 3$. The operation cost of this algorithm is $O(mn)$. Additionally,  \cite{kal} describes an iterative  algorithm that
in $t \leq n$ iterations, each of cost $O(mn)$,  produces an approximation of ${\rm diam}(S)$ to within a factor of  $c_*=\sqrt{5-2\sqrt{3}} \approx 1.24$.

In this paper we first prove that for $m=2$ it is possible to produce an approximation of diameter to within the factor of $c_*$ in $t=2$ iterations, thus giving an $O(n)$ approximation algorithm. In fact running this algorithm for general case of $m$ only $t=2$ iterations produces very good approximation for large test data.  Additionally,  we describe a simple randomized algorithm to approximate the diameter of a finite set of points in any dimension $m$.  This algorithm is a modified version of the algorithm presented in \cite{kal}.  We also test this algorithm for large data sets, making comparison with several algorithms in the literature. Our computational results demonstrate that the proposed algorithms here are superior in performance to the existing ones. The proposed algorithms appear to be extremely fast for a large variety of point distributions, in large dimensions. Moreover, these algorithms do not need to construct any
complicated data structure and very easy to implement. In addition to the memory required for the data, they only use constant memory.
The most relevant work to ours are those in \cite{mal,her} which we make  comparison to.

In  Section 2, we present an approximation algorithm described in \cite{kal} but prove that  in 2D it approximated the diameter to within a factor  of $c_*=\sqrt{5-2\sqrt{3}} \approx 1.24$. In Section 3, we formally describe  this algorithm for arbitrary dimension $m$, and give a
randomized version.  In Section 4, we present experimental results of the proposed algorithms in various dimensions and make comparison with several existing algorithms.

\section{A Fast approximation of Diameter in 2D} \label{sec1}

Let $S=\{p_1, \dots, p_n\}$ be a subset of $\mathbb{R} ^m$. We will first assume $m \geq 2$ is an arbitrary integer and describe an approximation algorithm but we will analyze the performance for $m=2$.  Let ${\rm diam} (S)$ be the diameter of $S$. Let $d( \cdot, \cdot)$ denote the Euclidean distance. Given $p \in \mathbb{R} ^m$, $r >0$, let
$B_r(p)=\{x \in \mathbb{R} ^m: d(x,p) \leq r\}$, the ball of radius $r$ centered at $p$. For a given point $p \in\mathbb{R} ^m$, let $f(p)$ denote the farthest point of $p$ in $S$. Let $r_p=d(p,f(p))$. We write $f^2(p)$ for $f(f(p))$.

Consider the following algorithm.  Pick arbitrary $p \in S$. Compute $f(p)$. Clearly, $S \subset B_{r_p}(p)$, see Figure \ref{fig0new}, and we have,
\begin{equation}
r_p\leq {\rm diam} (S) \leq 2 r_p.
\end{equation}

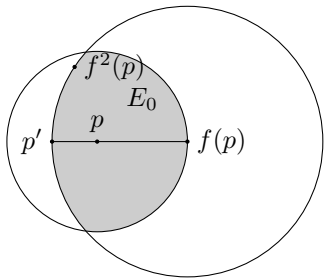
\begin{figure}[h]
	\centering
	\begin{tikzpicture}[scale=0.3]
\begin{scope}[red]
		 \clip  (4.0,0.0) circle (6);
		 \clip (0.0,0.0) circle (4);
\fill[color=gray!39] (-10, 10) rectangle (10, -10);
\end{scope}

		\begin{scope}[black]

		\draw (0.0,0.0) circle (4.0);
		\draw (4.0,0.0) circle (6);
		\end{scope}

		\filldraw (0,0) circle (2pt) node[above] {$p$};
\filldraw (2,1) node[above] {$E_0$};
		\filldraw (4,0) circle (2pt) node[right] {$f(p)$};
		\filldraw (-1,3.31) circle (2pt) node[right] {$f^2(p)$};
        \filldraw (-2,0) circle (2pt) node[left] {$p'$};

         \draw (-2.0,0.0) -- (4,0);
	\end{tikzpicture}
	
\begin{center}
\caption{$E_0$, an initial region containing $S$ based on two farthest point computations.} \label{fig0new}
\end{center}
\end{figure}

Next compute $f^2(p)$.  Let
\begin{equation}
E_0=B_{r_p}(p) \cap B_{r_{f(p)}}(f(p)).
\end{equation}
Clearly, $S \subset E_0$. The set $E_0$ is the intersection of two balls, forming an uneven eye-shape, see gray area in Figure \ref{fig0new}.  Its diameter gives a better factor bound than $2$. To estimate the diameter of $E_0$ we include it in a larger eye-shape region whose diameter can be estimated conveniently. Set
\begin{equation} \label{pprime}
p'=\alpha p + (1- \alpha) f(p), \quad \alpha = \frac{r_{f(p)}}{r_p}.
\end{equation}
This point $p'$ lies on the affine line joining $p$ and $f(p)$, a distance of $r_{f(p)}$ from $f(p)$. Let
\begin{equation}
E=B_{r_{f(p)}}(p') \cap B_{r_{f(p)}}(f(p)).
\end{equation}
The eye-shape region $E$ is the intersection of two balls
passing through each other's centers, see gray area in Figure \ref{fig1new}. The diameter of $E$ is known to be $\sqrt{3} r_{f(p)}$, see \cite{kal}.  From this it follows that
\begin{equation}
r_{f(p)} \leq {\rm diam} (S) \leq \sqrt{3} r_{f(p)}.
\end{equation}

The  diameter of $E$ is attained as the distance between the two corners of this eye-shape.  Let $c_1$ and $c_2$ be these corners, see Figure \ref{fig1new}. Clearly the complexity to obtain this $\sqrt{3}$-approximation to diameter is $2mn$ arithmetic operations.

To improve this bound, in \cite{kal} the following iterative procedure is described:  Let $q$ be the midpoint of $p'$ and $f(p)$, see Figure \ref{fig2new} for a case in the Euclidean plane. Specifically, from \ref{pprime} we have
\begin{equation} \label{q}
q = \frac{1}{2}p'+ \frac{1}{2} f(p) =\frac{\alpha}{2} p +(1- \frac{\alpha}{2}) f(p), \quad \alpha = \frac{r_{f(p)}}{r_p}.
\end{equation}

Compute $f(q)$ and $f^2(q)$.
If $d(f(q),f^2(q)) \leq d(f(p),f^2(p))$, then
\begin{equation}
{\rm diam} (S) \leq  c_* d(f(p),f^2(p)), \quad
c_*={\sqrt{5-2\sqrt{3}}} \approx 1.24.
\end{equation}
Otherwise, replaces $S$ with $S \setminus \{p,f(p)\}$, and repeat the process,
replacing $p$ with $q$, $f(p)$ with $f(q)$. That is, let $q'$ be the point on the line segment $qf(q)$ a distance of $r_{f(q)}/2$ from $f(q)$. Then compute $f(q')$ and $f^2(q')$ and checks if $d(f(q'),f^2(q')) \leq d(f(q),f^2(q))$. If so,  ${\rm diam} (S) \leq  c_* d(f(q),f^2(q))$. If not,  iterates again. Eventually, in $t \leq n$ iterations, each of cost $O(mn)$, we obtain an approximation of ${\rm diam} (S)$ to within a factor of $c_*$.  However, in \cite{kal} no constant bound on $t$ is given.  In the forgoing arguments we prove that when $m=2$,
\begin{equation}
{\rm diam} (S) \leq  c_* \max \bigg \{ d(f(p),f^2(p)), d(f(q),f^2(q)) \bigg \}.
\end{equation}

Hence in at most $8n$ operations, the cost of computing the farthest point of $4$ points, namely $f(p)$, $f^2(p)$, $f(q)$, $f^2(q)$, we have an approximation of diameter to with a factor of $c_*$. We thus improve the results in \cite{kal} for $m=2$.

We now proceed to prove this result. Having picked an arbitrary point $p \in S$,  we compute $f(p)$ and $f^2(p)$, and $p'$ as described above. Let $q$ be the midpoint of $p'$ and $f(p)$, see Figure \ref{fig2new} for a case in the Euclidean plane. Compute $f(q)$. Let
\begin{equation}
\rho_*=\frac{c_*}{2}= \frac{\sqrt{5-2\sqrt{3}}}{2} \approx .62.
\end{equation}
If
\begin{equation}
r_{q}  \leq \rho_* r_{f(p)},
\end{equation}
then
\begin{equation}
{\rm diam} (S) \leq 2r_{q} \leq c_* r_{f(p)}.
\end{equation}
This is obvious since
$r_{q} \leq \rho_* r_{f(p)} $ implies $S$ is contained in $B_{r_q}(q)$.
So we assume $r_{q} > \rho_* r_{f(p)}$.

Without loss of generality assume that
\begin{equation}
p'=(-\frac{1}{2},0, \dots, 0) \in \mathbb{R} ^m,  \quad f(p)=(\frac{1}{2},0, \dots, 0) \in \mathbb{R} ^m.
\end{equation}

Thus $q=0$ is the midpoint of the line segment $p'f(p)$. We may also assume that the corner points of $E$ are located at
\begin{equation}
c_1=(0, \dots, 0, \frac{\sqrt{3}}{2}) \in \mathbb{R} ^m,  \quad c_2=(0, \dots, 0, \frac{\sqrt{3}}{2}) \in \mathbb{R} ^m.
\end{equation}

Let $E_{1}$ and $E_{2}$ be defined as the two halves of $E$ determined as the intersection of $E$ with the orthogonal bisecting hyperplane to the line $c_1c_2$ (the upper and lower parts of $E$). We thus have $r_{f(p)}=1$. Let
\begin{equation}
r=r_{q}=d(q,f(q)).
\end{equation}
We assume  $r > \rho_*$ (since otherwise ${\rm diam}(S) \leq c_*$). To improve the bound on diameter we  compute of $f^2(q)$, the farthest point of $f(q)$ in $S$.   Let
\begin{equation}
d=d(f(q), f^2(q)).
\end{equation}
We first prove that if $d \geq 1$ the following holds
\begin{equation}
d \leq {\rm diam} (S) \leq c_* d.
\end{equation}

To prove this, consider first the case of Euclidean plane, see Figure \ref{fig2new}. Assume without loss of generality $f(q)$ has nonnegative coordinates.

\begin{figure}[h]
	\centering
	\begin{tikzpicture}[scale=0.3]
\begin{scope}[red]
		 \clip  (4.0,0.0) circle (6);
		 \clip (-2.0,0.0) circle (6);
\fill[color=gray!39] (-10, 10) rectangle (10, -10);
\end{scope}

		\begin{scope}[black]

		\draw (0.0,0.0) circle (4.0);
		\draw (4.0,0.0) circle (6);
		\draw (-2.0,0.0) circle (6);
		\end{scope}

\filldraw (2,1) node[above] {$E$};
		\filldraw (0,0) circle (2pt) node[above] {$p$};
		\filldraw (4,0) circle (2pt) node[right] {$f(p)$};
		\filldraw (-1,3.31) circle (2pt) node[right] {$f^2(p)$};
        \filldraw (-2,0) circle (2pt) node[left] {$p'$};
         \filldraw (1,5.195) circle (2pt) node[above] {$c_1$};
         \filldraw (1,-5.195) circle (2pt) node[below] {$c_2$};

         \draw (-2.0,0.0) -- (4,0);
	\end{tikzpicture}
	
\begin{center}
\caption{$E$, a region that contains $S$ with diameter bounded by $\sqrt{3} {\rm diam}(S)$.} \label{fig1new}
\end{center}
\end{figure}
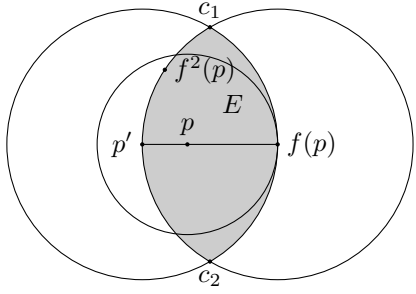

Let $q_*=(x_*,y_*)$ be the intersection of the two circles
\begin{equation}
(x-\frac{1}{2})^2+y^2=1, \quad x^2+y^2=r^2,
\end{equation}
where $x_* <0$, and $y_* >0$. This gives
\begin{equation}
x_*=r^2- \frac{3}{4} <0.
\end{equation}
We have
\begin{equation} \label{ystar}
y_*^2=r^2-x_*^2=r^2-(r^3-\frac{3}{4})^2= -r^4+\frac{5}{2}r^2- \frac{9}{16}.
\end{equation}
So
\begin{equation}
y_*=\sqrt{-r^4+ \frac{5}{2}r^2- \frac{9}{16}}.
\end{equation}
Let $w$ be the solution to the intersection of the sphere of radius $d$ centered at $f(q)$ and the sphere $(x+\frac{1}{2})^2+y^2=1$. Thus
\begin{equation}
d(f(q), w)=d.
\end{equation}
Note that we must have $d(q,w) \leq d(q,f(q))$.  Let $\overline q_*=(-x_*, y_*)$. Let $q'_*=(u,v)$ be the solution  to
\begin{equation}\label{eqx}
(x+\frac{1}{2})^2+y^2=1, \quad (x+ x_*)^2+(y-y_*)^2=d^2,
\end{equation}
where $v \leq 0$, (see Figure \ref{fig2new}). Thus
\begin{equation}
d=d(\overline q_*, q_*').
\end{equation}

Let $r'=d(q,q_*')$. Thus $u^2+v^2=r'^2$.

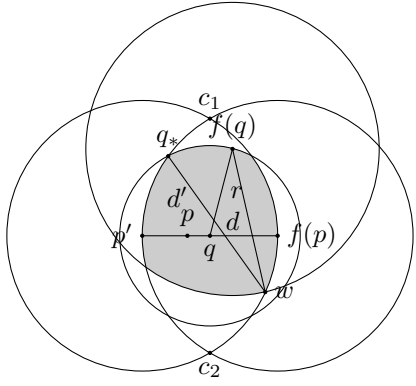
\begin{figure}[h]
	\centering
	\begin{tikzpicture}[scale=0.3]
		
\begin{scope}[red]
		 \clip  (4.0,0.0) circle (6);
		 \clip (-2.0,0.0) circle (6);
\clip (1.0,0.0) circle (4);
\clip (2,3.85) circle (6.5);
\fill[color=gray!39] (-10, 10) rectangle (10, -10);
\end{scope}

		\begin{scope}[black]

		\draw (4.0,0.0) circle (6);
		\draw (-2.0,0.0) circle (6);
\draw (1.0,0.0) circle (4);
\draw (2,3.85) circle (6.5);
		\end{scope}

		\filldraw (0,0) circle (2pt) node[above] {$p$};
		\filldraw (4,0) circle (2pt) node[right] {$f(p)$};
        \filldraw (-2,0) circle (2pt) node[left] {$p'$};
         \filldraw (1,5.195) circle (2pt) node[above] {$c_1$};
         \filldraw (1,-5.195) circle (2pt) node[below] {$c_2$};

         \filldraw (1,0) circle (2pt) node[below] {$q$};
         \filldraw (2,3.85) circle (2pt) node[above] {$f(q)$};
         \filldraw (3.45,-2.5) circle (2pt) node[right] {$w$};
         \filldraw (-.85,3.53) circle (2pt) node[above] {$q_*$};
         \draw (2,3.85) -- (3.45,-2.5) node[pos=0.5, left] {$d$};
         \draw (-.85,3.53) -- (3.45,-2.5) node[pos=0.3, left] {$d'$};
         \draw (1,0) -- (2,3.85) node[pos=0.5, right] {$r$};
         \draw (-2.0,0.0) -- (4,0);
		
	\end{tikzpicture}
	
\begin{center}
\caption{The region $E \cap B_r(q) \cap B_{d}(f(q))$, having diameter $d'$.} \label{fig2new}
\end{center}
\end{figure}

Let
\begin{equation}
d_*=d( q_*, q_*'), \quad d'=d(q_*,w).
\end{equation}

\begin{figure}[h]
	\centering
	\begin{tikzpicture}[scale=0.3]
		
\begin{scope}[red]
		 \clip  (4.0,0.0) circle (6);
		 \clip (-2.0,0.0) circle (6);
\clip (1.0,0.0) circle (4);
\clip (2.85,3.53) circle (6.732);
\fill[color=gray!39] (-10, 10) rectangle (10, -10);
\end{scope}

		\begin{scope}[black]

		\draw (4.0,0.0) circle (6);
		\draw (-2.0,0.0) circle (6);
\draw (1.0,0.0) circle (4);
\draw (2.85,3.53) circle (6.732);
		\end{scope}

		\filldraw (0,0) circle (2pt) node[above] {$p$};
		\filldraw (4,0) circle (2pt) node[right] {$f(p)$};
        \filldraw (-2,0) circle (2pt) node[left] {$p'$};
         \filldraw (1,5.195) circle (2pt) node[above] {$c_1$};
         \filldraw (1,-5.195) circle (2pt) node[below] {$c_2$};

         \filldraw (1,0) circle (2pt) node[below] {$q$};
         \filldraw (2,3.85) circle (2pt) node[above] {$f(q)$};
         \filldraw (3.45,-2.5) circle (2pt) node[right] {$w$};
         \filldraw (3.05,-3.2) circle (2pt) node[left] {$q'_*$};
         \filldraw (-.85,3.53) circle (2pt) node[above] {$q_*$};
        \filldraw (2.85,3.53) circle (2pt) node[right] {$\overline q_*$};
         \draw (2.85,3.53) -- (3.05,-3.2) node[pos=0.4, above] {$~~d$};
         \draw (-.85,3.53) -- (3.05,-3.2) node[pos=0.3, above] {$d_*$};
         \draw (2,3.85) -- (3.45,-2.5) node[pos=0.5, left] {$d$};
         \draw (1,0) -- (2,3.85) node[pos=0.5, right] {$r$};
         \draw (1,0) -- (3.05,-3.2)  node[pos=0.5, left] {$r'$};
         \draw (-2.0,0.0) -- (4,0);
		
	\end{tikzpicture}
	
\begin{center}
\caption{The region $E \cap B_r(q) \cap B_{d}(\overline q_*)$, having diameter $d_*$.} \label{fig3new}
\end{center}
\end{figure}

We first  prove some lemmas.

\begin{lemma} \label{lem1} Suppose $r' \leq r$ and  $d  \geq 1$. Then ${\rm diam}(E \cap B_r(q) \cap B_d(f(q)))=d'$.
\end{lemma}

\begin{proof} Consider a pair of points $u$, $v$ in the corresponding domain such that ${\rm diam}(E \cap B_r(q) \cap B_d(f(q)))=d(u,v)$, see Figure \ref{fig2new}. We claim that  $u$ and $v$ must be the extreme points of the domain (corner points).  Otherwise, assuming that $u,v$ do not coincide with $p', f(p)$, consider the line segment $uv$. Then the line that is orthogonal to this line segment, either at $u$ or at $v$, must cut through the feasible region.  But this means $d(u,v)$ is not the claimed diameter since it can be increased within the feasible domain.  On the other hand if  $u,v$ does coincide with $p', f(p)$, then we can change $f(p)$ to a corner point on its bounding circle until it touches another circle in which case we can replace it with a corner point. Now once one of $u$ or $v$ is a corner point, say $u$, we can consider the orthogonal line to $uv$ at $v$. Again this line must cut through the feasible region, contradicting that $d(u,v)$ is diameter. Once we have argued that both $u,v$ are extreme points we can consider different pairs and since we have assumed that $f(q)$ has nonnegative coordinates it follows that the diameter is $d'=d(q_*,w)$ as claimed.
\end{proof}

\begin{lemma} \label{lem2} Suppose $r' \leq r$ and  $d  \geq 1$. Then
${\rm diam}(E \cap B_r(q) \cap B_d(\overline q_*))=d_*$.
\end{lemma}

\begin{proof} For this consider Figure \ref{fig3new}. The proof of this lemma is analogous to the proof of previous one.
\end{proof}

\begin{lemma} \label{lem3} Suppose $r' \leq r$ and  $d  \geq 1$.
Then $ d \leq {\rm diam} (S) \leq d'  \leq d_*$.
\end{lemma}

\begin{proof} The first two inequalities are obvious. To prove the last inequality we only need to observe that  $w$ is feasible to $E \cap B_d(\overline q_*)$ (even though $E \cap B_r(q) \cap B_d(\overline q_*)$ is not a subset of  $E \cap B_r(q) \cap B_d(f(q))$). Figure \ref{fig4new} gives a superposition of the regions.
\end{proof}

Next we state our main result.

\begin{thm} Suppose $r' \leq r$ and  $d  \geq 1$. Then,  ${d_*}\leq c_* d$.
\end{thm}
\begin{proof}
Note that
\begin{equation}
d^2_*= (u- x_*)^2+(v-y_*)^2.
\end{equation}
Also $u^2+v^2=r'^2$. From the first equation in (\ref{eqx}) we get
\begin{equation}
u=\frac{3}{4}-r'^2.
\end{equation}
Note that
\begin{equation}
v^2=r'^2-u^2=r'^2-(\frac{3}{4}-r'^2)^2= -r'^4+\frac{5}{2}r'^2- \frac{9}{16}.
\end{equation}
So
\begin{equation}
v= - \sqrt{-r'^4+\frac{5}{2}r'^2- \frac{9}{16}}.
\end{equation}
Then from the second equation in (\ref{eqx}) we get
\begin{equation}
(u^2+v^2)+ (x_*^2+y_*^2)+2ux_*- 2vy_*=d^2.
\end{equation}
Since $u^2+v^2=r'^2$ and  $x_*^2+y_*^2=r^2$ we get
\begin{equation}
2(x_*u-y_*v)=d^2-r^2-r'^2.
\end{equation}
We wish to bound the ratio $d_*/d$ where
\begin{equation}
d^2_*=  (u- x_*)^2+(v-y_*)^2.
\end{equation}
Note that we have
\begin{equation}
d_*^2=r^2+r'^2-2x_*u-2y_*v, \quad d^2=r^2+r'^2+2x_*u-2y_*v.
\end{equation}
So we get
\begin{equation}
d^2_*= d^2 -4x_*u.
\end{equation}
Substituting for $x_*$ and $u$ and dividing by $d^2$ we get
\begin{equation}
\frac{d^2_*}{d^2}= 1 + \frac{4}{d^2}(\frac{3}{4}-r^2)(\frac{3}{4}-r'^2).
\end{equation}
Note that
\begin{equation}
d^2=r'^2+r^2- 2(\frac{3}{4}-r^2)(\frac{3}{4}-r'^2)+2
 \sqrt{-r^4+ \frac{5}{2}r^2- \frac{9}{16}}\sqrt{-r'^4+\frac{5}{2}r'^2-\frac{9}{16}}.
 \end{equation}
Let $a=r^2$, $b=r'^2$. Since $r' \leq r$, we wish to bound the following ratio in the region $1/4 < b \leq a \leq 3/4$:
\begin{equation}
\frac{4(\frac{3}{4}-a)(\frac{3}{4}-b)}{
a+b- 2(\frac{3}{4}-a)(\frac{3}{4}-b)+2
 \sqrt{-a^2+ \frac{5}{2}a- \frac{9}{16}}\sqrt{-b^2+\frac{5}{2}b-\frac{9}{16}}}.
 \end{equation}

Multiplying and dividing by $4$ and simplifying we get
\begin{equation}
\frac{(3-4a)(3-4b)}{
4a+4b- 0.5(3-4a)(3-4b)+0.5
 \sqrt{-16a^2+ 40a- 9}\sqrt{-16b^2+40b-9}}.
\end{equation}

Let $A=4a$ and $B=4b$.  Then the ratio becomes
\begin{equation}
F(A,B)= \frac{(3-A)(3-B)}{
A+B- 0.5(3-A)(3-B)+0.5
 \sqrt{-A^2+ 10A- 9}\sqrt{-B^2+10B-9}}, \quad 1 < B \leq A \leq 3.
\end{equation}
When $A=B$ we get
\begin{equation}
R(A)=F(A,A)=\frac{(3-A)^2}{
2A- 0.5(3-A)^2+0.5(-A^2+ 10A- 9)}=\frac{(3-A)^2}{
-A^2+10A-9}, \quad 1 <A \leq 3.
\end{equation}
The function $R(A)$ is monotonically decreasing on $[1+\epsilon, 3]$, $\epsilon>0$ arbitrarily small.
Using this, we can find the appropriate bound.  So it suffices to show that for fixed $B$,  $1 < B < A$, we have $F(A,B) \leq R(A)$. But this can be verified by calculus (also graphing it shows this).

Since $d_*/d \leq  \sqrt{1+R(A)}$, and $ 1 <A=4r^2 \leq 3$, we find a value for $r$ so that
\begin{equation}
2r= \sqrt{1+R(A)}= \sqrt{1+\frac{(3-A)^2}{
-A^2+10A-9}}.
\end{equation}
Equivalently, squaring the above we get
\begin{equation}A=1+R(A).\end{equation}
Solving this we get $A=5-2 \sqrt{3} \approx 1.5359$. Then
$r= \sqrt{5-2\sqrt{3}}/2 \approx .62$, and
\begin{equation}
\sqrt{1+R(A)}= \sqrt{5-2\sqrt{3}} \approx 1.24.
\end{equation}
\end{proof}

\begin{figure}[h]
	\centering
	\begin{tikzpicture}[scale=0.3]
		
		\begin{scope}[black]

		\draw (4.0,0.0) circle (6);
		\draw (-2.0,0.0) circle (6);
\draw (1.0,0.0) circle (4);
\draw (2,3.85) circle (6.5);
\draw (2.85,3.53) circle (6.732);
		\end{scope}

		\filldraw (0,0) circle (2pt) node[above] {$p$};
		\filldraw (4,0) circle (2pt) node[right] {$f(p)$};
        \filldraw (-2,0) circle (2pt) node[left] {$p'$};
         \filldraw (1,5.195) circle (2pt) node[above] {$c_1$};
         \filldraw (1,-5.195) circle (2pt) node[below] {$c_2$};

         \filldraw (1,0) circle (2pt) node[below] {$q$};
         \filldraw (2,3.85) circle (2pt) node[above] {$f(q)$};
         \filldraw (3.45,-2.5) circle (2pt) node[right] {$w$};
         \filldraw (3.05,-3.2) circle (2pt) node[left] {$q'_*$};
         \filldraw (-.85,3.53) circle (2pt) node[above] {$q_*$};
\filldraw (2.85,3.53) circle (2pt) node[right] {$\overline q_*$};
\draw (1,0) -- (3.05,-3.2)  node[pos=0.5, left] {$r'$};
         \draw (3.45,-2.5) -- (-.85,3.53);
         \draw (1,0) -- (2,3.85);
         \draw (3.45,-2.5) -- (1,0);
         \draw (2,3.85) -- (3.45,-2.5);
         \draw (-.85,3.53) -- (3.05,-3.2);
         \draw (3.05,-3.2) -- (2.85,3.53);
          \draw (-.85,3.53) -- (2.85,3.53);

         \draw (-2.0,0.0) -- (4,0);
		
	\end{tikzpicture}
	
\begin{center}
\caption{Superposition of  $E$, $B_r(q)$, $B_{d}(f(q))$,  and $B_{d}(\overline q_*)$.} \label{fig4new}
\end{center}
\end{figure}
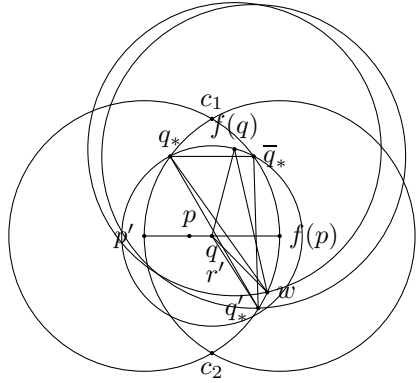

\begin{thm} Suppose $r' > r$ and  $d  \geq 1$. Then, $2y_* \leq d \leq {\rm diam}(S) \leq 2 r$. In particular,
$ {\rm diam}(S) \leq c_* d$.
\end{thm}
\begin{proof}  Since $r' >r$, $q_*'$ is outside of the ball of radius $r$ at $q$.  This means $d \geq 2 y_*$ (see definition of $y_*$,  (\ref{ystar})).  In other words $d$ is at least as long as the distance between $\overline q_*$ and its reflection with respect to the $x$-axis.  Now given the formula of $y_*$ in terms of $r$, (\ref{ystar}) we have
\begin{equation}
\frac{r}{y_*} = \frac{r}{-r^4+\frac{5}{2}r^2- \frac{9}{16}}.
\end{equation}
As a function of $r$ the maximum in the interval $[\rho_*, \sqrt{3}/2]$ is attained at $\rho_*$ and is $c_*$. It decreases to one at the other end point of the interval.
\end{proof}

We now prove the result when $d \leq 1$.

\begin{thm} Suppose $d  \leq 1$. Then,  ${\rm diam}(S) \leq c_*$.
\end{thm}

\begin{proof} Since $S \subset E \cap B_r(q) \cap B_d(f(q)) \subset E \cap B_r(q) \cap B_1(f(q))$ and we have already proved the diameter
of $E \cap B_r(q) \cap B_1(f(q))$ is bounded above by $c_*$, the proof of the Euclidean plane is complete.
\end{proof}

We end this analysis by exhibiting a case of five points where the worst-case bound of $c_*$ is achievable, see Figure \ref{figlast}.

    \begin{figure}[h]
	\centering
	\begin{tikzpicture}[scale=0.3]
		
\begin{scope}[red]
		 \clip  (4.0,0.0) circle (6);
		 \clip (-2.0,0.0) circle (6);
\clip (1.0,0.0) circle (4);
\fill[color=gray!39] (-10, 10) rectangle (10, -10);
\end{scope}

		\begin{scope}[black]
		\draw (4.0,0.0) circle (6);
		\draw (-2.0,0.0) circle (6);
\draw (1.0,0.0) circle (4);
		\end{scope}
		
		\filldraw (4,0) circle (2pt) node[right] {$f(p)$};
        \filldraw (-2,0) circle (2pt) node[left] {$p=p'$};

         \filldraw (1,0) circle (2pt) node[below] {$q$};
         \filldraw (2.85,3.53) circle (2pt) node[above] {$f(q)$};
         \filldraw (2.85,-3.53) circle (2pt) node[below] {$\overline q_*$};
         \filldraw (-.85,3.53) circle (2pt) node[above] {$q_*$};
         \draw (2.85,3.53) -- (2.85,-3.53) node[pos=0.4, right] {$d$};
         \draw (-.85,3.53) -- (2.85,-3.53) node[pos=0.3, above] {$d'=d_*$};
         \draw (1,0) -- (2.85,3.53) node[pos=0.5, right] {$r$};
         \draw (-2.0,0.0) -- (4,0);
		
	\end{tikzpicture}
	
\begin{center}
\caption{A case of five points gives worst-case error when $r=\rho_*$.} \label{figlast}
\end{center}
\end{figure}

\section{Iterative Algorithms }
In this section we propose two iterative algorithms for approximating the diameter of a finite set in any dimension $m \geq 2$.  The first algorithm is essentially identical with the algorithm in \cite{kal}.  The second is a randomized version of the algorithm.  We formally describe these in Algorithm 1 and Algorithm 2. They each have an input $t$ as the number of iterations, however we only give implementation of them for small $t$ because they produce high accuracy solution (within absolute error of $10^{-4}$ on tested data sets). We present experimental results with these algorithms.

\subsection{The iterative approximation algorithm}
In this section an iterative version of previous algorithm is presented. In each iteration given $p \in S$, we compute $f(p)$, $f^2(p)$, $q$, $f(q)$ and $f^2(q)$.   Computing $f(p)$ and $f^2(p)$ requires at most $2mn$ operations.  Computing $q$ requires $O(m)$ operations (see \ref{q}).  An additional $2mn$ operations are needed to compute $f(q)$ and $f^2(q)$. Thus when $m$ is much smaller than $n$ each iteration requires $4mn$ operations. In our computational results we have only run the algorithm $t=2$ times.
In Algorithm \ref{alg1}, the iterative algorithm is shown formally.

\begin{algorithm}[htb]                      
\caption{The iterative approximation algorithm}          
\small
\label{alg1}                           
\begin{algorithmic}[]                    
\STATE \textit{Input}:
     \STATE \hspace{\algorithmicindent}  $S$: a set of $n$ points in $\mathbb{R} ^m$
     \STATE \hspace{\algorithmicindent}  $t$: number of iterations

     \STATE \textit{Output}:
     \STATE \hspace{\algorithmicindent} An approximation value of {\rm diam}(S)
     \STATE $d_{max}=0$, $i=0$.
     \STATE choose an arbitrary point $p\in S$.
 \WHILE {$i<t$}
  \IF {$r_p=d(p,f(p))>d_{max}$}
  \STATE $d_{max}=d(p,f(p))$
  \ENDIF
  \IF {$r_{f(p)}=d(f(p),f^2(p))>d_{max}$}
  \STATE $d_{max}=d(f(p),f^2(p))$
  \ENDIF
  \STATE Let $q =\frac{\alpha}{2} p +(1- \frac{\alpha}{2}) f(p)$, $\alpha = \frac{r_{f(p)}}{r_p}$.
  \IF {$d(f(q),f^2(q))>d_{max}$}
\STATE  $d_{max}=d(f(q),f^2(q))$
  \ENDIF
   \STATE i=i+1
   \STATE $p=f^2(q)$
 \ENDWHILE
 \RETURN $d_{max}$
 \end{algorithmic}
\end{algorithm}

\subsection{Randomized and approximation Algorithm}

In this section, we present a randomized version of  Algorithm 1.
We begin from an arbitrary point $p$ and compute $f(p)$.  Next we compute the midpoint of $p$ and $f(p)$. Let $q$ be this midpoint. Then we compute
$f(q)$ and  $f^2(q)$. Thus the maximum of $d(p,f(p))$, $d(q, f(q))$, $d(f(q), f^2(q))$ is a lower bound to the diameter of $S$.
We iterate this algorithm. In the next step, we can either begin from $f(p)$ or $f^2(q)$. To do so, we randomly choose $f(p)$ or $f^2(q)$ with equal probability. This becomes our new point. Then we compute the farthest point from the chosen point and compare the estimate of diameter of previous step with the new one. In Algorithm~\ref{alg2} we explain the algorithm formally.  In each iteration of the algorithm we need about $3mn$ operations in contrast with $4mn$ operations in Algorithm 2.  In practice we ran this with $t=2,3$ and $5$. We have implemented this algorithm on some data sets. In the next section the experimental results of this algorithm are shown.


\begin{algorithm}[htb]                      
\caption{The iterative randomized algorithm}          
\label{alg2}                           
\small
\begin{algorithmic}[]                    
\STATE \textit{Input}:
     \STATE \hspace{\algorithmicindent}  $S$: a set of $n$ points in $\mathbb{R} ^m$
     \STATE \hspace{\algorithmicindent}  $t$: number of iterations

     \STATE \textit{Output}:
     \STATE \hspace{\algorithmicindent} An approximation value of {\rm diam}(S)
     \STATE $d_{max}=0$, $i=0$.
     \STATE choose an arbitrary point $p\in S$.
 \WHILE {$i<t$}
  \IF {$d(p,f(p))>d_{max}$}
  \STATE $d_{max}=d(p,f(p))$
  \ENDIF
  \STATE Let $q$ be the middle point on the line connecting $p$ and $f(p)$
  \IF {$d(f(q),f^2(q))>d_{max}$}
\STATE  $d_{max}=d(f(q),f^2(q))$
  \ENDIF
   \STATE i=i+1
   \STATE with probability $0.5$ let $p=f(p)$ and with probability $0.5$ let $p=f^2(q)$
 \ENDWHILE
 \RETURN $d_{max}$
 \end{algorithmic}
\end{algorithm}

	





	

\section{Experimental results}

To show the efficiency of the proposed algorithms in practice, we have implemented these and run them on some data sets.
The most comparable approaches to ours are the algorithms proposed in \cite{mal} and \cite{her}. We have used the package implemented by Malandain and Boissonnat's  in \cite{mal}. They have implemented their algorithms  and  we  have also implemented our algorithms and added  them to their package.
In their experiments they generated 2 types of data set:  Volume based distributions, in a cube, in a ball, and in sets of constant width (only in 2D); and Surface based distributions, on a sphere, and on ellipsoids. They also used  real inputs\footnote{ Large Geometric Models Archive,
http://www.cs.gatech.edu/projects/large models/, Georgia Institute of Technology.}. We have also used the same package to generate data sets and the same real inputs.
Malandain and Boissonnat's have implemented the following algorithms in \cite{mal}:
\begin{itemize}
\item Malandain and Boissonnat's  exact algorithm;
\item Malandain and Boissonnat's approximation algorithm;
\item Har-Peled's algorithm: implemented by Malandain and Boissonnat;
\item Hybrid1 algorithm: proposed by Malandain and Boissonnat which is combination of their algorithm and Har-Peled's algorithm;
\item Hybrid2 algorithm: another modification of the two algorithms presented by Malandain and Boissonnat's algorithm and Har-Peled's algorithm presented by Malandain and Boissonnat's.
\end{itemize}
We have generated the data sets and computed the diameter for each set using each of the above algorithms and our proposed algorithms.
The experimental results are shown in the tables.
The first 5 algorithms are implemented by Malandain and Boissonnat's and the next one is the implementation of Algorithm~\ref{alg1} with $t=2$ iterations. The next 3 are implementation of Algorithm~\ref{alg2} with $t=2, 3$ and $5$ iterations.


We make some observations in our computation.  Firstly, in all the data sets, the difference between the approximated value and exact value of diameter is less than $10^{-4}$ where ${\rm diam}(S)>1$ even with $t=2$ iterations for both algorithms.  From the tables it is seen that the running time of the randomized algorithm, Algorithm \ref{alg2} with $t=2$ iteration is better than all the other algorithms. The proposed algorithms are more efficient in higher dimensions.  One advantage of the proposed algorithms is that no extra memory is needed. Also, by virtue of their efficiency these algorithms can be implemented for big data sets.
Another advantage of these algorithms is that in higher dimensions, the running time of these algorithm is significantly better than the other algorithms (see Table \ref{hd}).
    \begin{table}[h]
\small
\caption{CPU times in millisecond for real data sets}
\begin{center}
\begin{tabular}{|c|c|c|c|c|}
\hline
Input&blade&dragon&hand&happy \\ \hline
Exact&70.9&8.06&17&12.38\\\hline
Approx&67&5.9&15&6.27\\ \hline
Har-Peled&84&84.6&36.8&76\\\hline
Hybrid1&112&102&13&64.2\\\hline
Hybrid2&42&79&13&54\\\hline
Algorithm \ref{alg1}&25.4&12.2&8.31&14.3\\\hline
Algorithm \ref{alg2}-5&34&16.9&14.2&20.5\\\hline
Algorithm \ref{alg2}-3&20.4&10.2&7.94&12.3\\\hline
Algorithm \ref{alg2}-2&1.35&6.9&5.3&8.2\\\hline
\end{tabular}
\end{center}
\label{sphereellipsoid}
\end{table}

\begin{table}[h]
\small
\caption{CPU time in milliseconds 3D Volume and Surface Based distributions
￼}
\begin{center}
\begin{tabular}{|c|c|c|c|c|c|c|}
\hline
Input&Cube&Cube&Cube&Ball&Ball&Ball\\
Points&10,000& 100,000&1,000,000&10,000&100,000&200,000 \\ \hline
Exact&1.2&6&72&0.23&1.88&7.43\\\hline
Approx&0.3&0.7&41&0.2&1.97&3.9\\ \hline
Har-Peled&2&16&322&0.39&10&41.70\\\hline
Hybrid1&2&17&183&1.5&16.9&63.11\\\hline
Hybrid2&1.94&16&163&1.4&27.4&62.11\\\hline
Algorithm \ref{alg1}&0.25&3.1&30.9&0.32&2.68&5.3\\\hline
Algorithm \ref{alg2}-5&0.34&3.6&38&0.52&3.981&7.574\\\hline
Algorithm \ref{alg2}-3&0.21&2.31&26.32&0.214&2.19&4.557\\\hline
Algorithm \ref{alg2}-2&0.18&1.49&17.48&0.148&1.48&3.08\\\hline
  \hline
Input&Sphere&Sphere&Sphere&Ellipsoid&Ellipsoid rotated &Ellipsoid regular\\
Points&10,000& 100,000&200,000&1,000,000&1,000,000&1,000,000 \\ \hline
Exact&0.29&3.03&5.9&100.28&52.81&51.86\\\hline
Approx&0.27&2.8&5.7&106.08&51.92&52.53\\ \hline
Har-Peled&122.71&7.96&33.27&284.69&103.88&113.8\\\hline
Hybrid1&1.29&0.75&3.22&95.19&66.65&66.93\\\hline
Hybrid2&1.43&0.71&2.94&0.114&65.1&68.75\\\hline
Algorithm \ref{alg1}&0.27&2.6&6.1&23.5&26.84&27.8\\\hline
Algorithm \ref{alg2}-5&0.34&3.59&7.69&44.3&38.65&38.89\\\hline
Algorithm \ref{alg2}-3&0.23&2.17&4.56&25.89&23.91&23.39\\\hline
Algorithm \ref{alg2}-2&0.14&1.53&3.66&15.23&15.77&15.44\\\hline
\end{tabular}
\end{center}
\label{cubeball}
\end{table}

\begin{table}[h]

\caption{CPU times in millisecond for synthetic distributions in higher dimensions.}
\begin{center}
\begin{tabular}{|c|c|c|c|c|}

\hline
  \multicolumn{5}{|c|}{Cube, n=100,000} \\
  \hline
Input&d=6&d=9&d=12&d=15 \\ \hline
Exact&73.24&235.57&1520&27520\\\hline
Approx&64.63&162.47&3210&16990\\ \hline
Har-Peled&207.52&24020&119890&138240\\\hline
Hybrid1&87.23&180.31&17960&15450\\\hline
Hybrid2&87.59&171.43&11800&12230\\\hline
Algorithm \ref{alg1}&5.3&7.2&8.5&10.9\\\hline
Algorithm \ref{alg2}-5&7.33&9.62&12.14&14.55\\\hline
Algorithm \ref{alg2}-3&4.41&5.77&7.34&8.99\\\hline
Algorithm \ref{alg2}-2&2.99&3.89&4.85&5.9\\\hline
\hline

\hline
  \multicolumn{5}{|c|}{Ball, n=100,000} \\
  \hline
Input&d=6&d=9&d=12&d=15 \\ \hline
Exact&68770&111590&178430&264660\\\hline
Approx&18550&48600&105200&168390\\ \hline
Har-Peled&86680&102800&120200&137530\\\hline
Hybrid1&93940&136260&73620&222960\\\hline
Hybrid2&71350&100030&146520&218510\\\hline
Algorithm \ref{alg1}&5.5&6.42&8.12&9.46\\\hline
Algorithm \ref{alg2}-5&7.25&9.58&11.98&14.35\\\hline
Algorithm \ref{alg2}-3&4.38&5.87&7.29&8.64\\\hline
Algorithm \ref{alg2}-2&2.89&3.89&4.82&5.79\\\hline

\hline
  \multicolumn{5}{|c|}{Ellipse, n=100,000} \\
  \hline
Input&d=6&d=9&d=12&d=15 \\ \hline
Exact&58&1110&770&2760\\\hline
Approx&31.9&504&589&1300\\ \hline
Har-Peled&3040&109760&127280&134070\\\hline
Hybrid1&90&2280&3240&5580\\\hline
Hybrid2&98&2020&3900&5200\\\hline
Algorithm \ref{alg1}&5.21&6.9&7.28&9.5\\\hline
Algorithm \ref{alg2}-5&7.3&9.5&12.1&16.5\\\hline
Algorithm \ref{alg2}-3&4.39&5.79&7.2&8.6\\\hline
Algorithm \ref{alg2}-2&2.988&4.2&4.9&5.8\\\hline
\hline
  \multicolumn{5}{|c|}{Regular Ellipse, n=100,000} \\
  \hline
Input&d=6&d=9&d=12&d=15 \\ \hline
Exact&24&350&150&2280\\\hline
Approx&31.7&285&110.5&1440\\ \hline
Har-Peled&530&62010&123720&133850\\\hline
Hybrid1&71&308&3430&4070\\\hline
Hybrid2&72.9&268&4590&119.5\\\hline
Algorithm \ref{alg1}&5.06&6.84&8.22&10.7\\\hline
Algorithm \ref{alg2}-5&7.3&11.2&12.2&17.8\\\hline
Algorithm \ref{alg2}-3&4.39&5.95&7.3&8.65\\\hline
Algorithm \ref{alg2}-2&2.952&4.05&4.8&5.8\\\hline

\end{tabular}
\end{center}
\label{hd}
\end{table}


\section{Conclusion}
In this paper, we studied computing the diameter of a point set in any dimension which is a significant problem in computational geometry. We presented a fast constant approximation factor algorithm, giving a worst-case bound of about $1.24$ in dimension $m=2$.  We believe that the same bound applies to any dimension. Its verification is the subject of future work. We proposed two iterative algorithms, one a randomized iterative algorithm.  We also implemented these algorithms and compared the running times with related works. Based on experimental results the algorithms are very efficient.  Deriving worst-case bounds for iterative algorithms in terms of the number of iterations $t$ remains as open problem.

\newpage


\clearpage
\begin{thebibliography}{4}

\bibitem{agr}
Amato, Nancy M., Michael T. Goodrich, and Edgar A. Ramos. ``Parallel algorithms for higher-dimensional convex hulls.'' Foundations of Computer Science, 1994 Proceedings., 35th Annual Symposium on. IEEE, 1994.
\bibitem{agha}
Agarwal, Pankaj K., Shankar Krishnan, Nabil H. Mustafa, and Suresh Venkatasubramanian. ``Streaming geometric optimization using graphics hardware.'' Algorithms-ESA 2003. Springer Berlin Heidelberg, 2003. 544-555.
\bibitem{bes}
Bespamyatnikh, Sergei. ``An efficient algorithm for the three-dimensional diameter problem.'' Discrete and Computational Geometry 25.2 (2001): 235-255.
\bibitem{cla}
Clarkson, Kenneth L., and Peter W. Shor. ``Applications of random sampling in computational geometry, II.'' Discrete and Computational Geometry 4.1 (1989): 387-421.
\bibitem{chan}
Chan, Timothy M. ``Faster core-set constructions and data stream algorithms in fixed dimensions.'' Proceedings of the twentieth annual symposium on Computational geometry. ACM, 2004.
\bibitem{kal}
Eg̃eciog̃lu, \"Omer, and Bahman Kalantari. ``Approximating the diameter of a set of points in the Euclidean space.'' Information Processing Letters 32.4 (1989): 205-211.
\bibitem{fou}
Fournier, Hervé, and Antoine Vigneron. ``A tight lower bound for computing the diameter of a 3D convex polytope.'' Algorithmica 49.3 (2007): 245-257.
\bibitem{her}
Har-Peled, Sariel. ``A practical approach for computing the diameter of a point set.'' Proceedings of the seventeenth annual symposium on Computational geometry. ACM, 2001.
\bibitem{mal}
Malandain, Grégoire, and Jean-Daniel Boissonnat. ``Computing the diameter of a point set.'' International Journal of Computational Geometry and Applications 12.06 (2002): 489-509.
\bibitem{prep}
Preparatat, Franco P., and Michael Ian Shamos. ``Computational geometry: an introduction.'' (1985).
\bibitem{ram1}
Ramos, Edgar A. ``Intersection of unit-balls and diameter of a point set in R< sup> 3</sup>.'' Computational Geometry 8.2 (1997): 57-65.
\bibitem{ram2}
Ramos, Edgar A. ``Deterministic algorithms for 3-D diameter and some 2-D lower envelopes.'' Proceedings of the sixteenth annual symposium on Computational geometry. ACM, 2000.
\bibitem{yao}
Yao, Andrew Chi-Chih. ``On constructing minimum spanning trees in k-dimensional spaces and related problems." SIAM Journal on Computing 11.4 (1982): 721-736.

\end{thebibliography}
\end{document}